\documentclass[aps,twocolumn,pra,10pt]{revtex4-1}
\usepackage{graphics}
\usepackage{helvet}
\usepackage{dcolumn}
\usepackage{bm}
\usepackage{amsmath}
\usepackage{hyperref}
\hypersetup{colorlinks=true, urlcolor=blue}
\usepackage{color}
\usepackage{amssymb}
\usepackage{amsfonts}
\usepackage{graphicx}
\usepackage{physics}
\usepackage{amsthm}

\newtheorem{claim}{Claim}
\newtheorem{theorem}{Theorem}
\newtheorem{lemma}{Lemma}

\numberwithin{tsubcase}{tcase}

\def\a{\alpha}
\def\b{\beta}
\def\g{\gamma}
\def\d{\delta}
\def\m{\mathit{m}}
\def\mB{\mathcal{B}}
\def\mS{\mathcal{S}}

\def\cc#1{c_{#1}^2}
\def\c#1{c_{#1}}
\def\be{\begin{equation}}
\def\ee{\end{equation}}

\begin{document}
\allowdisplaybreaks 
\addtolength{\jot}{1em} 

\title{Complementarity between tripartite quantum correlation and bipartite Bell inequality violation in three qubit states}

\author{Palash Pandya$^\ddag$, Avijit Misra$^\dag$, Indranil Chakrabarty$^\ddag$}
\affiliation{$^\ddag$ Center for Security Theory and Algorithmic Research,\\ 
International Institute of Information Technology, Gachibowli, Hyderabad, India}

\affiliation{$^\dag$ Harish Chandra Research Institute, Chhatnag Road, Jhunsi, Allahabad, UP, India.}

\begin{abstract}
 We find a single parameter family of genuinely entangled three qubit pure states, called the maximally 
 Bell inequality violating states (MBV), which exhibit maximum Bell inequality violation by the reduced  bipartite system for a fixed amount of genuine tripartite entanglement quantified by the so called tangle measure. This in turn implies that there holds a complementary relation between the Bell inequality violation by the reduced bipartite systems and the tangle present in the three qubit states, not necessarily pure. The MBV states also exhibit maximum Bell inequality violation by the reduced  bipartite systems of the three qubit pure states with a fixed amount of genuine tripartite correlation quantified by the generalized geometric measure, a genuine entanglement measure of multiparty pure states, and the discord monogamy score, a multipartite quantum correlation measure from information theoretic paradigm. The aforementioned complementary relation has also been established for three qubit pure states for the generalized geometric measure and the discord 
monogamy score respectively. The complementarity between the Bell inequality violation by the reduced bipartite 
systems and the genuine tripartite correlation suggests that the Bell inequality violation in the reduced two qubit system comes at the cost of the total tripartite correlation 
present in the entire system.
\end{abstract}

\maketitle

\section{ Introduction}
\label{sec:intro}
Quantum entanglement \cite{schrodinger, qent-rmp} allows correlation between distant parties that are completely forbidden in the classical regime \cite{epr-1935}. In 1964, John Bell formally showed that predictions of quantum mechanics are incompatible with the notion of local realism \cite{johnbell-epr}. 
After Bell's celebrated work there have been numerous conceptual and technical developments for studying nonlocality in quantum systems \cite{bell-rmp}. The Bell inequality and Bell-type inequalities \cite{johnbell-epr,chsh,bell-type,svetlichny} set an upper bound on the correlations between measurement
statistics of many particle systems that cannot be explained by local hidden variable models. If the measurement statistics violates some Bell-type inequality then the particles are said to posses nonlocal correlation \cite{incompleteness-nonlocality-redhead}. Nonlocality is not only of fundamental interest but it has also been employed as a key resource in several quantum information protocols such as key distribution \cite{key-dist} and quantum randomness generation \cite{qrgen}. On the other hand quantum entanglement has been a key resource for various information theoretic protocols such as dense coding \cite{dense,mdcc}, teleportation \cite{teleport} and many others \cite{protocols}. As quantum entanglement and nonlocality both are essential as resources in information theory it is interesting to establish the link between them. It is always exciting to know which states are both entangled and nonlocal and thus can be useful for several information theoretic protocols 
simultaneously. For example in Ref. \cite{vertesi}, it has been established that all quantum states useful for teleportation are nonlocal resources. All bipartite pure entangled states violate the Bell inequality and the magnitude of the violation is directly proportional to the amount entanglement of the states \cite{gisin-npstates}. 
In Ref. \cite{horodecki-mrho}, the necessary and sufficient condition for a two qubit mixed state to violate the Bell-CHSH inequality has been derived. Nonlocality and entanglement has been extensively studied in qubit systems and anisotropic one dimensional {\it XY} model in presence of a transverse magnetic field in \cite{batle-nl-ent} and \cite{batle-XY} respectively.
Nonlocality in more than two parties is much more complex than the two party scenario. Therefore, the link between the multiparty nonlocality and multiparty correlation is far more complex. Despite remarkable progress the precise link is still not clear in this regard \cite{bell-rmp}.

 In Ref. \cite{sghose-rungta}, the relation between Svetlichny's inequality \cite{svetlichny} and a specific genuine entanglement measure, tangle  has been studied
for families of three qubit pure states. One of the main problems in detecting genuine multiparty nonlocality in three qubit systems is to distinguish between the violations 
arising from reduced density matrices and from the tripartite state \cite{collins-cereceda}. However, in Ref. \cite{tura2013} it has been shown very recently, that in 
specific cases one or two body correlation functions are sufficient to detect nonlocality in many body systems. Moreover, several entanglement criteria have been proposed 
based on the two body expectation values \cite{criterion}. Considering the significance of the Bell inequality violations by reduced density matrices in composite systems, 
in this work we study how the Bell inequality violation by the reduced bipartite systems depends on the genuine tripartite correlation in three qubit pure 
states. We find that a single parameter class of genuinely entangled three qubit pure states, abbreviated as the maximally Bell inequality violating states (MBV), 
exhibits maximum Bell inequality violation by one of its reduced two qubit systems for a fixed amount of genuine tripartite correlation. This implies a complementary 
relation between the bipartite Bell inequality violation and the genuine tripartite correlation measures in three qubit pure states, with the MBV states lying at the 
boundary of the complementary relation. In this work, we consider the well known tangle measure \cite{Coffmankunduwotters}  and the generalized geometric measure (GGM), a genuine entanglement measure of multipartite pure
states \cite{ggm-ref,ggm} from the entanglement separability paradigm  and the discord monogamy score 
(DMS) \cite{discord,dms1} from the information theoretic paradigm. The complementary relation holds for all the three aforementioned correlation measures. As the $GHZ$ 
and the $W$ classes are two 
disjoint but complete subsets of genuinely entangled three qubit pure states, we consider the $GHZ$ and the $W$ classes separately to establish the aforementioned 
complementary relation. The $GHZ$ and $W$ classes of states are characterized by five and four independent parameters respectively, as we will discuss in Sec. \ref{sec:comp}.
We have established the complementarity analytically for the entire $W$ class of states, while for the $GHZ$ class of states we have kept one parameter fixed, i.e., we 
consider four parameters to establish the complementary relation analytically. However, from numerical study we have claimed that the complementary relation holds in general
for the set of three qubit pure states. Given the complementary relation between the maximum bipartite Bell inequality violation 
and the tangle holds for all three qubit pure states, it can be proved that it holds for all three qubit mixed states, by using the convexity properties of the maximum bipartite 
Bell inequality violation and the tangle. Thus, we claim that the complementary relation for the tangle holds 
for all three qubit states, not necessarily pure.
Our result can be used in a scenario where three parties share genuinely entangled systems to perform information theoretic protocols among them and at the same time 
they might need nonlocal resources between their subparts. In this regard it is very useful to know which state is more nonlocal in the bipartite scenario for a fixed 
amount of genuine tripartite correlation. 

 The organization of the paper is as follows. In Sec. \ref{sec:corr}, we provide a brief review of the measures 
of genuine quantum correlations that we have used in this work. We discuss the Bell-CHSH inequality and a No-go theorem for the same in Sec. \ref{sec:bell}.
In Sec. \ref{sec:comp}, we study the relation of the Bell inequality violation in the reduced subsystems with the genuine correlation present in the tripartite system and 
establish the complementarity between them. We summarize our result and conclude in Sec. \ref{conclu}.

\section{Various Measures of Tripartite Quantum Correlation}
\label{sec:corr}

 In this section we briefly discuss the different types of measures that we use in this work 
 for quantifying quantum correlation of the genuinely entangled three qubit pure states. By genuinely entangled it is meant that the quantum state under consideration is not separable in any bipartite  cut. These measures are \textit{A.}Tangle, \textit{B.} Generalized Geometric Measure (GGM), and \textit{C.} Discord Monogamy Score (DMS). Among them the first two belong to the entanglement-separability paradigm, while the third one belongs to the
information-theoretic paradigm. Moreover, the measures, tangle and discord monogamy score are based on the concept of monogamy of quantum correlations and the GGM is conceptualized using the geometric distance between two quantum states.

\subsection{ Tangle}
 The tangle is a genuine entanglement measure of three qubit systems that is conceptualized based on monogamy of quantum correlations \cite{Coffmankunduwotters,monogamybunch}.
The quantum monogamy score \cite{mdcc} corresponding to the square of the 
bipartite entanglement measure, called the Concurrence \cite{wootters-concurrence},  is the tangle \cite{Coffmankunduwotters}.
Concurrence of a two qubit system is given as
$C(\rho_{AB}) = \max\{0,\lambda_1-\lambda_2-\lambda_3-\lambda_4\},$
where $\lambda_1,\dots,\lambda_4$ are the square roots of the eigenvalues of $\rho_{AB}\tilde{\rho}_{AB}$ in 
decreasing order, $\tilde{\rho}_{AB}= (\sigma_y \otimes \sigma_y)\rho_{AB}^{*}(\sigma_y \otimes \sigma_y)$. 
Here the complex conjugation $\rho_{AB}^{*}$ is taken in the computational basis, and $\sigma_y$ is the Pauli spin matrix.

 Therefore, the tangle of a three-qubit state $\rho_{ABC}$ is given by \cite{Coffmankunduwotters}
\begin{equation}
\tau(\rho_{ABC}) = C^2_{A:BC} - C^2_{AB} - C^2_{AC}.
\label{eq:tangle}
\end{equation}

 The tangle is always non-negative \cite{Coffmankunduwotters}. In particular, it is zero only for the $W$-class states and positive for the $GHZ$-class states \cite{DurVidalCirac}.
\subsection{ Generalized Geometric Measure}
  
A multipartite pure quantum state $|\psi_{A_1,A_2,\ldots, A_N}\rangle$ is genuinely multipartite entangled if it is not separable across 
any  bipartition. The Generalized Geometric Measure (GGM) \cite{ggm} quantifies the genuine multipartite entanglement for these $N$-party states  
 based on the distance from the set of all multiparty states that are not genuinely entangled. The GGM is given as
\begin{equation}
\label{eq:GGM1}
{\cal G} (\psi_{A_1,A_2,\ldots, A_N}) = 1 - \max_{\ket{x}} |\bra{x}\ket{\psi_{A_1,A_2,\ldots, A_N}}|^2.
\end{equation}
This maximization is done over all states $|x\rangle$ which are not genuinely entangled. An equivalent mathematical expression of Eq.\eqref{eq:GGM1}, is the following
\begin{equation}
\label{GGM2}
\mathcal{G} (\psi_n) =  1 - \max \{\lambda^2_{ I : L} |  I \cup  L = \{A_1,\ldots, A_N\},  I \cap  L = \emptyset\},
\end{equation}
where \(\lambda_{I:L}\) is  the maximal Schmidt coefficient in the  bipartite split $I: L$ of $| \psi_N \rangle$.

\subsection{Discord Monogamy Score}
 The discord monogamy score is a genuine quantum correlation measure from the information theoretic paradigm, which is based on monogamy of quantum correlations \cite{monogamybunch}, as the name suggests. Quantum discord \cite{discord} is considered to be the the difference between the total correlation and the classical correlation of a two-party quantum state $\rho_{AB}$ and is given as
\begin{equation}
\label{eq:discord}
D(\rho_{AB})= S(\rho_B)+ S(\rho_{A|B})- S(\rho_{AB}),
\end{equation}
where $S(\rho)= - \mbox{tr} (\rho \log_2 \rho)$ is the von Neumann entropy of the quantum state \(\rho\). Here 
 \(\rho_A\) and \(\rho_B\) are the reduced density matrices of the quantum state  \(\rho_{AB}\).
$S(\rho_{A|B})$ is the measured conditional entropy when a projection-valued measurement is performed on $B$ and is given by
\begin{equation}
S(\rho_{A|B}) = \min_{\{P_i\}} \sum_i p_i S(\rho_{A|i}).
\end{equation}
The minimization is carried out over the complete set of rank-one projectors $\{P_i\}$ and  
$p_i = \mbox{tr}_{AB}[(\mathbb{I}_A \otimes P_i) \rho_{AB} (\mathbb{I}_A \otimes P_i)]$, is the probability of obtaining the outcome $P_i$. $\mathbb{I}_A$ is the identity operator on the Hilbert space of $A$.
The output state is $\rho_{A|i} = \frac{1}{p_i} \mbox{tr}_B[(\mathbb{I}_A \otimes P_i) \rho_{AB} (\mathbb{I}_A \otimes P_i)]$, corresponding to the outcome $P_i$.

 The quantum monogamy score corresponding to the quantum discord is the discord monogamy score \cite{dms1} which is given as
\begin{equation}
\delta_D = D_{A:BC} - D_{AB} - D_{AC}.
\label{eq:}
\end{equation}
Note that the quantum discord for three-qubit states can be both monogamous and non-monogamous unlike the square of the concurrence \cite{dms1}.

\section{Bell Inequality Violation}
\label{sec:bell}
 In 1964 John Bell established that violation of the Bell inequality by a two-party state excludes any local realistic description of the state. 
All bipartite pure entangled states violate the Bell inequality and thus forbid any local realistic description for them. The necessary and sufficient condition 
for a two-qubit mixed state to violate the Bell-CHSH inequality~\cite{chsh} was first given in Ref. \cite{horodecki-mrho}.
For an arbitary two-qubit state $\rho$, violation of the Bell-CHSH inequality implies that the Bell-CHSH value $\mathcal{S}$ is greater than $2$. It can be shown that the maximum Bell-CHSH value $\mathcal{S}(\rho)$ for a two-qubit state $\rho$ is given by \cite{horodecki-mrho}
\begin{equation}
\label{eq:max-chsh-violation}
 \mathcal{S}(\rho)=2\sqrt{M(\rho)},
\end{equation}
where $M(\rho)=m_1+m_2$, with $m_1$ and $m_2$ are the two largest eigenvalues of $T_{\rho}^T T_{\rho}$. $T_{\rho}$ is the correlation matrix associated with the two-qubit state $\rho$ with entries $(T_{\rho})_{ij}=Tr(\sigma_i \otimes\sigma_j\rho)$, where $\sigma_i$'s are the Pauli matrices. $T_{\rho}^T$ denotes the transpose of $T_{\rho}$.
Therefore, violation of the Bell-CHSH inequality implies
\begin{equation}\label{eqn:Mrho}
 M(\rho)>1. 
\end{equation}
Henceforth, by violation of the Bell inequality by a two-qubit state we mean violation of the Bell-CHSH inequality. Note that in this article, we consider the Bell inequality violation by two qubit quantum states only. 
The quantity $\mathcal{B}(\rho_{AB})$, defined as
\begin{equation}
\label{eq:biv-def}
\mathcal{B}(\rho_{AB})=\max\{0,M(\rho_{AB})-1\}\,,
\end{equation}
quantifies the amount of Bell inequality violation and hence the nonlocal correlation of two qubit quantum state $\rho_{AB}$ \cite{bell-monogamy}. 
Among the three reduced two qubit states of a three qubit pure state $\ket{\psi}$, to pick the one with the maximum Bell inequality violation we define
\begin{equation}
\mathcal{B}_{\max}(\psi) =\max\{\mathcal{B}(\rho_{AB}),\mathcal{B}(\rho_{BC}),\mathcal{B}(\rho_{AC})\}\,.
\end{equation}
In this context it is interesting to note that there are two qubit mixed entangled states which do not violate the Bell inequality \cite{horodecki-mrho}. 
In other words, violation of the Bell inequality by a quantum state is not synonymous with the idea of the state being entangled. It was shown that such a local character of quantum correlations traces back to the monogamy trade-off obeyed by bipartite Bell correlation \cite{non-violation}.
 Monogamy for the Bell inequality violation \cite{bell-monogamy} implies that, \textit{if a quantum state shared by three qubits leads to the violation of the Bell inequality among any two of its sub-parts, then it prohibits 
its violation for the other two states which the sub-parts share with the third party.} Monogamy of the Bell inequality violation, thus imposes general constraints on the nature of entanglement and Bell correlation \cite{non-violation}. 
In this paper we deal with the Bell inequality violations by the reduced two qubit systems of the three qubit pure states $\ket{\psi}_{ABC}$. Since only one among the three reduced density matrices of $\ket{\psi}_{ABC}$ can violate the Bell inequality, the bipartite Bell inequality
violation of $\ket{\psi}_{ABC}$ implies that it either comes from $\rho_{AB}$, $\rho_{AC}$ or $\rho_{BC}$. As a direct consequence, only one of $\mathcal{B}(\rho_{AB})$, $\mathcal{B}(\rho_{BC})$, or $\mathcal{B}(\rho_{AC})$ is non-zero, which is then picked up by the quantity $\mathcal{B}_{\max}(\psi)$.

\section{Bell inequality violation versus Quantum Correlation}
\label{sec:comp}
 In this section we establish a relation between the bipartite Bell inequality violation and the genuine tripartite correlation for three qubit pure states. In particular we show that there exists a complementary relation between the genuine tripartite quantum correlation measures and the bipartite Bell inequality violation of three qubit pure states. We identify the single parameter family of genuinely tripartite entangled three qubit pure states that give the maximum bipartite Bell inequality violation for a fixed amount of tripartite correlation, thus lying at the boundary of the aforesaid complementary relation for each of the three correlation measures that are considered in this work.
This single parameter family of the genuinely entangled three qubit pure states is given by
\begin{equation}
\label{mdcc-class-state}
\ket{\psi}_\m = \frac{\ket{000} + \m(\ket{010}+\ket{101})+ \ket{111}}{\sqrt{2+2\m^2}},
\end{equation}
where $m\in[0,1]$. These states belong to the $GHZ$ class  when $m\in[0,1)$. For $\m=1$, the state belongs to the $W$ class as
the tangle is zero at this point. We denote this class of states as the maximally Bell inequality violating (MBV) class of states. This class of states has been recognized as the maximally dense-coding capable (MDCC) \cite{mdcc}, for having maximal dense coding capabilities with fixed amount of genuine tripartite quantum correlations.
\subsection{Tangle versus the Maximum Bell Inequality Violation}
 We first derive the complementarity between the tangle and the bipartite Bell inequality violation for genuinely entangled three qubit pure states.
As the $GHZ$ and the $W$ classes are two disjoint but complete subsets of genuinely entangled three qubit pure states \cite{DurVidalCirac},
it is sufficient to establish the complementarity for $GHZ$ and $W$ classes individually.

 The $GHZ$ class is defined as a set of states which can be converted into the $GHZ$ state, $\frac{1}{\sqrt{2}}\left (\ket{000}+\ket{111}\right )$, using SLOCC with non-vanishing probability \cite{DurVidalCirac}.
This class of states is characterized by five parameters, $\a,\b,\g,\d$ and $\phi$, with a general state in the class being defined as follows
\begin{equation}
\label{ghz-class-state}
\ket{\psi}_{GHZ} = \sqrt{K} \left( c_{\d} \ket{000} + s_{\d} e^{i \phi} \ket{\varphi_\a}\ket{\varphi_\b}\ket{\varphi_\g} \right),
\end{equation}
where $c_\d$ and $s_\d$ stand for $\cos{\d}$ and $\sin{\d}$ respectively, $K = \left( 1+ c_\a c_\b c_\g c_\phi s_{2\d} \right)^{-1}$ is the normalizing constant, and
\begin{align}
\ket{\varphi_\a} &= c_\a\ket{0} + s_\a\ket{1} ,\nonumber \\ 
\ket{\varphi_\b} &= c_\b\ket{0} + s_\b\ket{1} ,\nonumber \\
\ket{\varphi_\g} &= c_\g\ket{0} + s_\g\ket{1} ,
\end{align}
with parameter ranges, $\a , \b , \g \in (0 , \pi/2]$, $\d \in (0, \pi / 4] $ and $\phi \in [0, 2\pi)$. 
When $\phi=0$, we denote the corresponding $GHZ$ subclass and the respective states by $GHZ^R$ and $\ket{\psi}_{GHZ^R}$ respectively. It is worth mentioning that the complementary relation is shown analytically, to hold for the $GHZ^R$ class of states for all the three correlation measures considered in this paper. However, numerical studies suggest that it holds for the entire $GHZ$ class of states for all the three correlation measures.

To  derive the complementary relation for the $GHZ^R$ class state, we begin 
with the following theorem.
\begin{theorem} \label{theorem3}
If two three qubit pure states $\ket{\psi}_{GHZ^R}$ and $\ket{\psi}_m$ are such
that the corresponding tangles, $\tau(\psi_{GHZ^R})$ and $\tau(\psi_{m})$ are equal,
then the maximal bipartite Bell inequality violations necessarily follow
\begin{equation} 
	\mathcal{B}_{\max}({\psi}_{m}) \geq \mathcal{B}_{\max}({\psi_{GHZ^R}}). 
\end{equation}
\end{theorem}

\begin{proof} The tangle for $\ket{\psi}_m$ is given by (see Appendix \ref{app:tangle}),
\begin{align}
\label{eq:mdcc-tangle}
\tau(\psi_m) = 1 -\frac{4m^2}{(1+m^2)^2}.
\end{align}
From Eq.\eqref{eq:biv-mbv} and Eq.\eqref{eq:mdcc-tangle}, it follows that
\begin{equation}
\label{eq:mtangle-bv}
\mathcal{B}_{\max}({\psi_\m}) +\tau(\psi_\m)= 1. 
\end{equation}

 The tangle for $\ket{\psi}_{GHZ^R}$ is given by 
\begin{equation}
\label{eq:gtangle2}
\tau(\psi_{GHZ^R}) = \frac{s^2_\a s^2_\b s^2_\g s^2_{2\d}}{\left (1+c_\a c_\b c_\g s_{2\d} \right )^2}.
\end{equation}

 As $\tau(\psi_\m) = \tau(\psi_{GHZ^R})$, the Bell inequality violation for $\ket{\psi_m}$ can be written as
\begin{equation}
\label{eq:tabvm}
\mathcal{B}_{\max}({\psi_\m}) = 1-\frac{s^2_\a s^2_\b s^2_\g s^2_{2\d}}{\left (1+c_\a c_\b c_\g s_{2\d} \right )^2}.
\end{equation}

 Let us consider the case when the reduced bipartition $\rho_{BC}$ violates the Bell-CHSH inequality, which implies $\mathcal{B}_{\max}(\psi_{GHZ^R}) = \mathcal{B}(\rho_{BC})>0$. The Bell inequality violation for the density matrix $\rho_{BC}$ of $\ket{\psi_{GHZ^R}}$ is given by
\begin{equation}
\label{eq:ghz-bv-bc1}
	\mathcal{B}(\rho_{BC}) = \frac{\left (c_\a^2 -c_\b^2 -c_\g^2 + 2c_\b^2 c_\g^2\right ) s^2_{2 \d } -\mathcal{X}^2 }{ (1+\mathcal{X})^2}, 
\end{equation}
where $\mathcal{X}=c_\a c_\b c_\g s_{2\d}$.
Now,
\begin{equation*}
\begin{split}
&\mathcal{B}\left ({\psi_\m}\right ) -\mathcal{B}\left (\rho_{BC}\right ) \\ 
		 =\ & 1-\frac{s^2_\a s^2_\b s^2_\g s^2_{2\d}+\left (c_\a^2+2c_\b^2c_\g^2-c_\b^2-c_\g^2\right )s_{2\d}^2 - \mathcal{X}^2}{\left (1+\mathcal{X}\right )^2}\\
		 =\ & 1+\frac{\mathcal{X}^2 -s_{2\d}^2 +\left (s_\a^2 c_\b^2 s_\g^2 + c_\g^2 s_\a^2+c_\b^2s_\g^2+c_\g^2s_\b^2\right )s_{2\d}^2 }{(1+\mathcal{X})^2}.
\end{split}
\end{equation*}

 Note that,  $\mathcal{B}_{\max}({\psi_\m}) -\mathcal{B}(\rho_{BC})
\geq 1 - s_{2\d}^2 \geq  0$.

 Similarly, when $\mathcal{B}_{\max}(\psi_{GHZ^R}) = \mathcal{B}(\rho_{AB})$ or $\mathcal{B}_{\max}(\psi_{GHZ^R}) = \mathcal{B}(\rho_{AC})$ (See Appendix \ref{app:bv} for the Bell inequality violation of $\rho_{AC}$ and $\rho_{AB}$), one can also show that if $\tau(\psi_\m) = \tau(\psi_{GHZ^R})$, then the following inequalities respectively hold 
\begin{equation*}
\begin{split}
\mathcal{B}_{\max}({\psi_\m}) &\geq\mathcal{B}(\rho_{AC}),  \\
\mathcal{B}_{\max}({\psi_\m}) &\geq\mathcal{B}(\rho_{AB}) \,. 
\end{split}
\end{equation*}

 Hence the relation, $\mathcal{B}_{\max}(\psi_{m}) \geq \mathcal{B}_{\max}(\psi_{GHZ^R})$ if $\tau(\psi_\m) = \tau(\psi_{GHZ^R}).$
\end{proof}

 Let us now derive the relation between the bipartite Bell inequality violation of an arbitrary $W$ class state $\ket{\psi}_W$ and $\ket{\psi}_m$.
Note that the $W$ class states have zero tangle, and for $\ket{\psi}_m$ the tangle, $\tau(\psi_m)=0$, if and only if $m=1$. In which case we denote $\ket{\psi}_m$ as $\ket{\psi}_1$. Next, we prove the following theorem for the $W$ class states.

\begin{theorem}\label{theorem4}
Among all three-qubit pure states belonging to the $W$ class, $\ket{\psi}_1$ exhibits the maximum bipartite Bell inequality violation, i.e.,
\begin{equation}
\mathcal{B}_{\max}(\psi_1) \geq \mathcal{B}_{\max}(\psi_{W}),
\end{equation}
where $\ket{\psi}_1$, is the MBV state $\ket{\psi}_m$ for $m=1$.
\end{theorem}
\begin{proof}
For any two qubit state $\rho_{XY}$, $\mathcal{B}(\rho_{XY})\leq1$ \cite{horodecki-mrho}.
Putting $m=1$ in Eq.\eqref{eq:biv-mbv}, we get $\mathcal{B}_{\max}(\psi_1)=1$, which implies that 
$\ket{\psi}_1$ shows the maximum bipartite Bell inequality violation among all $W$ class states.
 Hence the result, 
\begin{equation*}
\mathcal{B}_{\max}(\psi_1) \geq \mathcal{B}_{\max}(\psi_{W}). \qedhere
\end{equation*}
\end{proof}

\begin{figure}
\begin{center}
\includegraphics[scale=0.6]{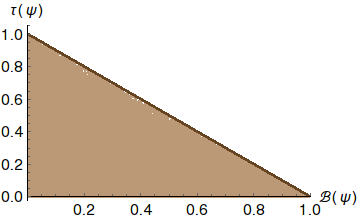}
\end{center}
\caption{Complementarity between the tangle, $\tau(\psi)$ and the maximum Bell inequality violation, $\mathcal{B}_{\max}(\psi)$ for $10^6$ number of Haar uniformly generated random three qubit pure states. The MBV states lies at the boundary. Both axes are dimensionless.}
\label{fig:ghz-mdcc-tangle}
\end{figure}

 From Theorem \ref{theorem3}. and \ref{theorem4}., for a three qubit pure state $\ket{\psi}$  either from the $GHZ^R$ or the $W$ class, we have $\mathcal{B}_{\max}(\psi) < \mathcal{B}_{\max}(\psi_m)$ when $\tau(\psi)=\tau(\psi_m)$.
Therefore, for the states in $GHZ^R$ and $W$ classes, one has the following complementary relation
\begin{align}
\label{tangle-comple-pure}
  \tau(\psi) + \mathcal{B}_{\max}(\psi)\leq 1. 
\end{align}

 In Fig.\ref{fig:ghz-mdcc-tangle}, the complementary relation between the tangle and the bipartite Bell inequality violation has been studied for Haar uniformaly 
generated random three qubit pure states. The numerical evidence suggests that the complementary relation holds in general for all the three qubit pure states.
From numerical evidences, we claim that:

 \textit{For \textbf{an arbitrary three qubit pure state} $\ket{\psi}$, the tangle $\tau(\psi)$, and the maximal bipartite Bell inequality
violation $\mathcal{B}_{\max}(\psi)$, follow the
 complementary relation given in Eq. \ref{tangle-comple-pure}.}

 {\bf Extension to mixed states:}
In what follows, we will prove that if the complementary relation between the tangle and the maximum bipartite Bell inequality violation holds for all three qubit pure 
states, then it also holds for three qubit mixed states. We will use the convexity property of the tangle and maximum bipartite Bell inequality violation by the reduced two
qubit states  respectively, of an arbitrary three qubit mixed states.

\begin{proof}
The tangle of mixed three qubit states has been defined by convex roof construction \cite{tangle-mixed}, i.e., the tangle of a mixed three qubit quantum state $\rho$ is 
given as
\begin{equation}
 \label{mixed-tangle}
 \tau(\rho)= \min_{\{p_i, |\psi_i\rangle\}} \sum_i p_i \tau(\psi_i),
\end{equation}
where the minimization is over the all pure state decomposition of $\rho$, such that $\rho=\sum_i p_i \op{\psi_i}$, where $p_i\geq 0$ and $\sum_i p_i=1.$ The convexity property of the tangle of mixed states 
follows from the definition itself. Now, we will prove the convexity property of the maximum bipartite Bell inequality violation $\mathcal{B}_{\max}(\rho)$ for three 
qubit mixed states.  The maximum Bell-CHSH value \cite{horodecki-mrho} for the reduced two qubit system $\rho_{AB}=\Tr_C [\rho_{ABC}]$ of a three qubit mixed state $\rho_{ABC}$ 
is given by
\begin{align}
 \mS(\rho_{AB}&) = \left |\Tr[(O_{AB}^{m}\otimes \mathbb{I}_C)\rho_{ABC}] \right| \nonumber \\
  				=& \left |\sum_i p_i\left (\Tr[\left (O_{AB}^{m}\otimes \mathbb{I}_C\right ) |\psi_{ABC}^i \rangle \langle \psi_{ABC}^i|]\right )\right |\,. 
\end{align}
Here, $O_{AB}^{m}$ is the optimal Bell-CHSH operator \cite{horodecki-mrho} for $\rho_{AB}$, $\mathbb{I}_C$ is the identity operator on the Hilbert space of $C$, and $\rho_{ABC}=\sum_i p_i|\psi_{ABC}^i \rangle \langle \psi_{ABC}^i|$ is an arbitrary pure state decomposition of $\rho_{ABC}$. Now using the subadditivity property of absolute value we have
\begin{align}
 \mS(\rho_{AB}) \leq \sum_i p_i&\left |\Tr[\left ( O_{AB}^{m} \otimes \mathbb{I}_C\right )|\psi_{ABC}^i \rangle \langle \psi_{ABC}^i|]\right |\nonumber\\
 =& \sum_i p_i\left |\Tr[ O_{AB}^{m}\rho_{AB}^i]\right |,
\end{align}
where $\rho_{AB}^i= \Tr_C[|\psi_{ABC}^i \rangle \langle \psi_{ABC}^i|]$. 
In the RHS of the above inequality, $O_{AB}^{m}$ is an non-optimal Bell-CHSH operator for the states $\rho_{AB}^i$'s. If it is replaced by $O_{AB}^{i,m}$ that is optimal 
for $\rho_{AB}^i$ i.e., gives the maximum Bell-CHSH value, for $\rho_{AB}^i$ for all $i$, then we have 
%
\begin{align}
 \label{eq:conv-chsh}
\mS(\rho_{AB})\leq  \sum_i p_i\left |\Tr[ O_{AB}^{i,m}\rho_{AB}^i]\right |\,, \nonumber \\
 \mS\left (\sum_ip_i\rho_{AB}^i\right) \leq \sum_i p_i\mS(\rho_{AB}^i). 
\end{align}
Similar relations can be shown for the other two subsystems $\rho_{AC}$ and $\rho_{BC}$ as well. 
Following Eq. \eqref{eq:max-chsh-violation}, the Ineq. \eqref{eq:conv-chsh} can be written as
\begin{align}
\label{eq:conv1}
\sqrt{M\left (\sum_i p_i\rho_{AB}^i\right )} \leq \sum_i p_i \left( \sqrt{M(\rho_{AB}^i)} \right )\,.
\end{align}
Note that $M(\rho)$ for a two qubit system $\rho$ is defined as in Eq. \eqref{eq:max-chsh-violation}. It is necessary to mention that Eq. \eqref{eq:conv1} is derived for a two qubit system $\rho_{AB}$ which is a convex combination of two qubit states $\rho_{AB}^i$'s, where $\rho_{AB}^i= \Tr_C[|\psi_{ABC}^i \rangle \langle \psi_{ABC}^i|]$. Therefore, $\rho_{AB}^i$'s are rank two states as $|\psi_{ABC}^i \rangle$'s are three qubit pure states. However, note that Eq. \eqref{eq:conv1} holds true for any convex combination $\rho_{AB}=\sum_jq_j\rho_{AB}^j$ ($q_j\geq 0$, $\sum_jq_j=1$), of an arbitrary two qubit state $\rho_{AB}$ as we have only used linearity property of trace and sub-additivity property of absolute value in deriving the Eq. \eqref{eq:conv1}. Therefore, Eq. \eqref{eq:conv1} implies convexity of the function  $\sqrt{M(\rho)}$. 
As square of a convex non-negative function is still convex \cite{conv-analysis}, it follows that
\begin{equation}
\label{eq:conv-mrho}
M\left (\sum_i p_i\rho_{AB}^i\right )\leq  \sum_i p_i \left[ M(\rho_{AB}^i) \right ]\,.
\end{equation}

Subtracting 1 from both sides we get
\begin{equation}
\label{eq:conv4}
M\left (\sum_i p_i\rho_{AB}^i\right )-1 \leq  \sum_i p_i [ M(\rho_{AB}^i)-1]\,.
\end{equation}
Hence, it follows that
\begin{align}
\label{eq:conv5}
\max\{0,M(\rho_{AB})-1\} \leq&  \max\{0, \sum_i p_i [ M(\rho_{AB}^i)-1]\} \nonumber\\ \leq& \sum_i p_i \max\{0,[ M(\rho_{AB}^i)-1]\}\nonumber\\ \mB (\rho_{AB}) \leq& \sum_i p_i\mB(\rho_{AB}^i)\,.
\end{align}
Similarly,
\begin{align}
 \mB (\rho_{BC}) \leq& \sum_i p_i\mB(\rho_{BC}^i), \\
 \mB (\rho_{AC}) \leq& \sum_i p_i\mB(\rho_{AC}^i).
\end{align}

Now, the maximum bipartite Bell Inequality violation, $\mB_{\max}(\rho_{ABC}) = \max\{ \mB_{AB}(\rho_{ABC}),\mB_{AC}(\rho_{ABC}),$ $ \mB_{BC}(\rho_{ABC})\}$.
Without loss of generality say $\mB_{\max}(\rho_{ABC})=\mB(\rho_{AB})$.
Again, 
\begin{align}
 \mB (\rho_{AB}) \leq& \sum_i p_i\mB(\rho_{AB}^i) \nonumber\\
 \leq&\sum_i p_i[\max\{\mB(\rho_{AB}^i),\mB(\rho_{BC}^i),\mB(\rho_{AC}^i)\}] \nonumber\\
 \leq&\sum_i p_i[\mB_{\max}(|\psi_{ABC}^i\rangle)].
\end{align}
Therefore, 
\begin{equation}
  \mB_{\max}(\rho_{ABC})\leq \sum_i p_i[\mB_{\max}(|\psi_{ABC}^i\rangle)].
\end{equation}
Hence, the maximum bipartite Bell inequality violation of a three qubit mixed state is convex.
%

 As the tangle and the maximum bipartite Bell inequality violation of a three qubit mixed state both are convex it implies that the complementary relation in Eq. \eqref{tangle-comple-pure} holds for all three qubit mixed states if it holds for all three qubit pure states. 
\end{proof}

Therefore, from the support of numerical study for entire $GHZ$ class states, we claim in the following that the complementary relation holds for all three qubit states, not necessarily pure.

\begin{claim}\label{claim-tangle-mixed}
For an arbitrary three-qubit state $\rho$, the tangle $\tau(\rho)$ and the maximum bipartite Bell Inequality violation $\mB_{\max}(\rho)$ follows the following complementary relation 
\begin{equation}
\tau(\rho) + \mB_{\max}(\rho) \leq 1.
\end{equation}
\end{claim}

\subsection{GGM versus Maximum Bell Inequality Violation} 
 Let us now derive the complementarity 
between the GGM and the bipartite Bell inequality violation for the $GHZ^R$ and $W$ class states. 
\begin{lemma}\label{lemma1}
If for a three qubit pure state $\ket{\psi}_{GHZ^R}$, the GGM is obtained from, say the $A:BC$ bipartite split, then the only 
reduced bipartite system of $\ket{\psi}_{GHZ^R}$ that can violate the Bell inequality is $\rho_{BC}$.
\end{lemma}
\begin{proof}
 The GGM of the state $\ket{\psi}_{GHZ^R}$, is given as $\mathcal{G}(\psi_{GHZ^R}) = 1-\max\{g_A,g_B, g_C\}$,
where $g_i$ is the maximum eigenvalue of the reduced state $\rho_i$ of $\ket{\psi}_{GHZ^R}$, with $i \in \{A,B,C\}$. 
These maximum eigenvalues, $g_i$'s, are given by (see Appendix \ref{app:ggm}),
\begin{align}
\label{ghz-l-a}
g_A = \frac{1}{2}\left (1 + \sqrt{1+\frac{ s_\a^2(c_\b^2 c_\g^2 -1)  s_{2\d}^2}{ (1+\mathcal{X})^2 }}\right ), \\
\label{ghz-l-b}
g_B = \frac{1}{2}\left (1 + \sqrt{1+\frac{ s_\b^2(c_\a^2 c_\g^2 -1)  s_{2\d}^2}{ (1+\mathcal{X})^2 }}\right ), \\
\label{ghz-l-c}
g_C = \frac{1}{2}\left (1 + \sqrt{1+\frac{ s_\g^2(c_\a^2 c_\b^2 -1)  s_{2\d}^2}{ (1+\mathcal{X})^2 }}\right ), 
\end{align}
where, $\mathcal{X}=c_\a c_\b c_\g s_{2\d}$. Let us suppose, without loss of generality,
that the GGM is obtained from the bipartite split $A:BC$. 
Consequently, $\max\{ g_A,g_B,g_C\} = g_A $ and $\mathcal{G}(\psi_{GHZ^R}) = 1-g_A$.

 Now, if $g_A>g_B$ we obtain the following condition
\begin{align}
s_\a^2\left (c_\b^2 c_\g^2 -1\right )&\geq  s_\b^2\left (c_\a^2 c_\g^2 -1\right ) \nonumber \\
\Rightarrow \left (c^2_\a - c^2_\b\right )s^2_\g &\geq 0, 
\label{eq:lemma1-cond1}
\end{align}
and similarly $g_A>g_C$ implies
\begin{equation}
\left (c^2_\a - c^2_\g\right )s^2_\b \geq 0.
\label{eq:lemma1-cond2}
\end{equation}

 The Bell-CHSH values $M(\rho_{BC})$, $M(\rho_{AC})$ and $M(\rho_{AB})$ for the reduced states $\rho_{BC}$, $\rho_{AC}$ and $\rho_{AB}$ of $\ket{\psi}_{GHZ^R}$ (see Appendix \ref{app:bv}), are given as
\begin{align}
M(\rho_{BC}) = 1+\frac{(c_\a^2 -c_\b^2 -c_\g^2 + 2c_\b^2 c_\g^2) s^2_{2 \d } -\mathcal{X}^2 }{ (1+\mathcal{X})^2} ,\\
M(\rho_{AC}) = 1+\frac{(c_\b^2 -c_\a^2 -c_\g^2 + 2c_\a^2 c_\g^2) s^2_{2 \d } -\mathcal{X}^2 }{ (1+\mathcal{X})^2} ,\\
M(\rho_{AB}) = 1+\frac{(c_\g^2 -c_\a^2 -c_\b^2 + 2c_\a^2 c_\b^2) s^2_{2 \d } -\mathcal{X}^2 }{ (1+\mathcal{X})^2}.
\end{align}
Now using Eq.\eqref{eq:lemma1-cond1}, we obtain
\begin{equation}
M(\rho_{BC})-M(\rho_{AC})=  \frac{2\left ( c^2_\a - c^2_\b\right) s_\g^2 s_{2\d}^2}{(1+\mathcal{X})^2} \geq 0.
\label{eq:lemma1-diff1}
\end{equation}
Similarly, Eq.\eqref{eq:lemma1-cond2} implies
\begin{equation}
M(\rho_{BC})-M(\rho_{AC})=  \frac{2\left ( c^2_\a - c^2_\g\right) s_\b^2 s_{2\d}^2}{(1+\mathcal{X})^2} \geq 0.
\label{eq:lemma1-diff2}
\end{equation}

It follows from the no-go theorem for the Bell Inequality violation that if $\rho_{BC}$ violates the Bell-CHSH inequality, 
then $M(\rho_{BC})>1$, $M(\rho_{AC})\leq 1$ and $M(\rho_{AB})\leq 1$. Hence, from Eq.\eqref{eq:lemma1-diff1} and Eq.\eqref{eq:lemma1-diff2}, we get $\mathcal{B}(\rho_{BC})> 0$ and $\mathcal{B}(\rho_{AC})=\mathcal{B}(\rho_{AB})=0$.

 Note that cyclic permutations of the variables $\a$, $\b$ and $\g$ enables one to obtain, say $g_B$ and $g_C$ from $g_A$.
 Therefore, a similar proof holds for the cases when the GGM is obtained from the other two bipartite splits. This completes the proof of the lemma.
\end{proof}

 While we have proved Lemma \ref{lemma1}. for $GHZ^R$ class, a numerical check for $10^6$ Haar uniformly generated random $GHZ$ class states indicates that this lemma holds for all the $GHZ$ class states.
\begin{theorem} \label{theorem1}
If $\ket{\psi}_{GHZ^R}$ and $\ket{\psi}_m$ be two three qubit pure states such that the corresponding GGM, $\mathcal{G}(\psi_{GHZ^R})$ and $\mathcal{G}(\psi_m)$ are equal then the maximal bipartite Bell inequality violations for the two states necessarily follow

\begin{equation} 
	\mathcal{B}_{\max}(\psi_m) \geq \mathcal{B}_{\max}(\psi_{GHZ^R}) .
\end{equation}

\end{theorem}

\begin{proof}
 The GGM (Appendix \ref{app:ggm}) and the maximum Bell inequality violation of the MBV state $\ket{\psi}_m$ (Appendix \ref{app:bv}), are given as
\begin{align}
\mathcal{G}(\psi_m) =\frac{1}{2} - \frac{m}{1+m^2}, \\
\mathcal{B}_{\max}(\psi_m) = \frac{4m^2}{(1+m^2)^2}, \label{eq:biv-mbv}
\end{align}
\begin{equation}
\label{eq:mggm-bv}
\mathcal{B}_{\max}(\psi_\m)  = 4\left (\frac{1}{2}-\mathcal{G}(\psi_\m)\right )^2.
\end{equation}

 Let us assume that for $\ket{\psi}_{GHZ^R}$, $\max\{ g_A,g_B,g_C\} = g_A$, then the GGM is obtained from $A:BC$ bipartite split and
\begin{multline}
\mathcal{G}(\psi_m) = \mathcal{G}(\psi_{GHZ^R}) = 1-g_A \\ = \frac{1}{2} \left(1 - \sqrt{1+\frac{ s_\a^2(c_\b^2 c_\g^2 -1)  s_{2\d}^2}{ (1+\mathcal{X})^2 }}\right).
\end{multline}
Then from Eq.\eqref{eq:mggm-bv}, we have
\begin{equation}
\mathcal{B}_{\max}(\psi_\m)_{A:BC} = 1 + \frac{  (c_\a^2 +c_\b^2 c_\g^2 -1) s_{2\d}^2 -\mathcal{X}^2 }{(1+\mathcal{X})^2},
\end{equation}
where the subscript $A:BC$ indicates that the GGM of $\ket{\psi}_{GHZ^R}$ is obtained from the bipartite split $A:BC$.
 From Lemma \ref{lemma1}., it follows that if GGM is obtained from $A:BC$ split, then only $\mathcal{B}(\rho_{BC})$ can be nonzero and thus $\mathcal{B}_{\max}(\psi_{GHZ^R}) = \mathcal{B}(\rho_{BC})$.
Therefore, to prove theorem \ref{theorem1}., we only have to show $\mathcal{B}_{\max}(\psi_m)_{A:BC}\geq\mathcal{B}(\rho_{BC})$.

 Now,
\begin{equation*}
\begin{split}
   &\mathcal{B}_{\max}(\psi_\m)_{A:BC} - \mathcal{B}(\rho_{BC}) \\
=\ & 1+ \frac{(c_\b^2 s_\g^2 - s_\g^2)s_{2\d}^2}{(1+\mathcal{X})^2} 
=\  1- \frac{s_\b^2 s_\g^2 s_{2\d}^2}{(1+\mathcal{X})^2}\geq 0,
\end{split}
\end{equation*}
as $\mathcal{X}\geq 0$.

Similarly, it can be proved that when $\mathcal{G}(\psi_{GHZ^R})=1-{g}_B$, or $\mathcal{G}(\psi_{GHZ^R})=1-{g}_C$, 
we have $\mathcal{B}_{\max}(\psi_\m)_{B:AC} \geq \mathcal{B}(\rho_{AC})$ or $\mathcal{B}_{\max}(\psi_\m)_{C:AB} \geq \mathcal{B}(\rho_{AB})$ respectively.

 Hence, $\mathcal{B}_{\max}(\psi_m)\geq\mathcal{B}_{\max}(\psi_{GHZ^R})$, when $\mathcal{G}(\psi_m) = \mathcal{G}(\psi_{GHZ^R})$.
\end{proof}

 Now we derive the complementary relation between the Bell inequality violation and GGM of the $W$ class states.
The $W$ state, $\frac{1}{\sqrt{3}}(\ket{001}+\ket{010}+\ket{100})$, is the representative for the $W$ class which is a set of all states that can be converted to the $W$ state using SLOCC with non-zero probability. 
The $W$ class states are given by
\begin{equation}
\label{w-class-state}
\ket{\psi}_W = \left (\sqrt{d}\ket{000} + \sqrt{a}\ket{001} + \sqrt{b}\ket{010} + \sqrt{c}\ket{100}\right ),
\end{equation}
where $a + b + c + d=1$ is the normalizing condition, and $a,b,c>0$, $d=1-(a+b+c)\geq0$ \cite{DurVidalCirac}. To establish the complementary relation let us 
 first prove the following lemma for $W$ class states.

\begin{lemma} \label{lemma2}
If for a three qubit pure state $\ket{\psi}_W$, the GGM is obtained from, say $A:BC$ split, then the only reduced bipartite system of $\ket{\psi}_W$ that can violate the Bell inequality is $\rho_{BC}$.
\end{lemma}

\begin{proof}
The maximum eigenvalues corresponding to the reduced systems $\rho_A$, $\rho_B$ and $\rho_C$ for $\ket{\psi}_W$ are respectively given by (see Appendix \ref{app:ggm}),
\begin{align}
w_A = \frac{1}{2}\left (1 + \sqrt{1- 4(a+b)c}\right ) ,\\
w_B = \frac{1}{2}\left (1 + \sqrt{1- 4(a+c)b}\right ) ,\\
w_C = \frac{1}{2}\left (1 + \sqrt{1- 4(b+c)a}\right ) ,
\end{align}
where $w_X$ is the maximum eigenvalue of the subsystem $\rho_X$ of state $\ket{\psi_W}$, with $X\in\{A,B,C\}$.
If we assume that the GGM is obtained from the bipartite split $A:BC$, then it implies that $w_A>w_B$ and $w_A>w_C$.

 From $w_A\geq w_B$ and $w_A\geq w_C$ we respectively get
\begin{align}
\label{eq:lemma2-cond1}
a \left ( b-c \right ) &\geq 0 , \\
\label{eq:lemma2-cond2}
b \left ( a-c \right ) &\geq 0 .
\end{align}

 For the reduced states $\rho_{BC}$, $\rho_{AC}$ and $\rho_{AB}$ of $\ket{\psi}_W$, the Bell-CHSH values $M(\rho_{BC})$, $M(\rho_{AC})$ and $M(\rho_{AB})$ are 
(see appendix C) respectively given as
\begin{align}
M(\rho_{BC}) = \frac{1}{2} + \frac{ 12a b -4 a c - 4 b c + \sqrt{V}}{ 2 }, \\
M(\rho_{AC}) = \frac{1}{2} + \frac{ 12a c -4 a b - 4 b c + \sqrt{V}}{ 2 }, \\
M(\rho_{AB}) = \frac{1}{2} + \frac{ 12b c -4 a b - 4 a c + \sqrt{V}}{ 2 },
\end{align}
where, $V = [(\sqrt{a}+\sqrt{b}-\sqrt{c})^2+d][(\sqrt{a}-\sqrt{b}+\sqrt{c})^2+d]  [(-\sqrt{a}+\sqrt{b}+\sqrt{c})^2 +d][(\sqrt{a}+\sqrt{b}+\sqrt{c})^2+d]$.

 Then from Eq.\eqref{eq:lemma2-cond1} and  Eq.\eqref{eq:lemma2-cond2}, one can show
\begin{align}
\label{eq:lemma2-diff1}
M(\rho_{BC})-M(\rho_{AC}) = 8a \left ( b-c \right ) &\geq 0 ,\\
\label{eq:lemma2-diff2}
M(\rho_{BC})-M(\rho_{AB}) = 8b \left ( a-c \right ) &\geq 0 .
\end{align}
The two equations above demonstrate that $M(\rho_{BC})$  is larger among the three. It follows from the no-go theorem for Bell inequality violation, that when $M(\rho_{BC})>1$, then $M(\rho_{AC})\leq1$ and $M(\rho_{AB})\leq1$. 
Therefore, it follows from Eq.\eqref{eq:biv-def}, that only $\mathcal{B}(\rho_{BC})>0$ and $\mathcal{B}(\rho_{AC})=\mathcal{B}(\rho_{AB})=0$, when $\mathcal{G}(\psi_W)=1-w_A$. Permutation of parameters $a$, $b$ and $c$ gives rise 
to the cases when the GGM is obtained from the other two bipartite splits, and follow a similar proof. This completes the proof of the lemma.
\end{proof}
  From the Lemma \ref{lemma1}., Lemma \ref{lemma2}. and the numerical studies for $GHZ$ class states we 
conjecture that for all entangled three qubit pure states, 
if the bipartite split $X:YZ$ gives the GGM, then among the three reduced bipartite states, only $\rho_{YZ}$ can exhibit the  Bell inequality violation.

\begin{theorem}\label{theorem2}
If $\ket{\psi}_W$ and $\ket{\psi}_m$ are two three qubit pure states such that the respective GGM, $\mathcal{G}(\psi_W)$ and $\mathcal{G}(\psi_m)$ are equal, then the maximal bipartite Bell inequality violations necessarily follow
\begin{equation}
\mathcal{B}_{\max}(\psi_m)\geq\mathcal{B}_{\max}(\psi_W).
\end{equation}
\end{theorem}
\begin{proof}
 Let us first assume, that $\max\{w_A,w_B,w_C\} = w_A$ and hence, $\mathcal{G}(\psi_W) = 1-w_A$. Therefore,
\begin{multline}
\mathcal{G}(\psi_m) = \mathcal{G}(\psi_W) = 1-w_A \\= \frac{1}{2} \left (1- \sqrt{1-4(a+b)c}\right ).
\end{multline}
Employing Eq.\eqref{eq:mggm-bv}, the corresponding maximum Bell inequality violation for $\ket{\psi}_m$ is given as
\begin{align}
\mathcal{B}_{\max}(\psi_\m&)_{A:BC} = 1 - 4(a+b) c,
\end{align}
with the subscript indicating that the bipartite split $A:BC$ gives the GGM for $\ket{\psi}_W$.

 From Lemma \ref{lemma2}., it follows that only $\mathcal{B}(\rho_{BC})\geq0$, and hence, $\mathcal{B}_{\max}(\psi_{GHZ^R})=\mathcal{B}(\rho_{BC})$. 
Thus, we only need to show that $\mathcal{B}_{\max}(\psi_m)_{A:BC}>\mathcal{B}(\rho_{BC})$:

Now,
\begin{equation*}
\begin{split}
   \mathcal{B}_{\max}(\psi_\m)_{A:BC} &- \mathcal{B}(\rho_{BC}) \\
=\ & \frac{3}{2}-\frac{4ac + 4bc +12ab + \sqrt{V}}{2}, \\
\end{split}
\end{equation*}
where $V = [(\sqrt{a}+\sqrt{b}-\sqrt{c})^2+d][(\sqrt{a}-\sqrt{b}+\sqrt{c})^2+d]  [(-\sqrt{a}+\sqrt{b}+\sqrt{c})^2 +d][(\sqrt{a}+\sqrt{b}+\sqrt{c})^2+d]$.

 It can be shown that the minimum of $\mathcal{B}_{\max}(\psi_\m)_{A:BC} - \mathcal{B}(\rho_{BC})$ is zero, and therefore $\mathcal{B}_{\max}(\psi_\m)_{A:BC} \geq \mathcal{B}(\rho_{BC})$. 
The same can be proved for the cases when the other two bipartite splits give the GGM, following the cyclic order of the  parameters $a$, $b$ and $c$.
\begin{figure}[tbp!]
\begin{center}
\includegraphics[scale=0.6]{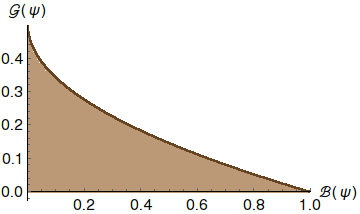}
\end{center}
\caption{Complementarity between the GGM, $\mathcal{G}(\psi)$ and the maximum Bell inequality violation, $\mathcal{B}_{\max}(\psi)$ for $10^6$ number of Haar uniformly generated random three qubit pure states. The MBV states lies at the boundary. Both axes are dimensionless.}
\label{fig:ggm-bv-ghz-w}
\end{figure}
\end{proof}

 From Theorem \ref{theorem1}. and \ref{theorem2}., for a state $\ket{\psi}$ belonging to either of the $GHZ^R$ or $W$ classes, we have  that $\mathcal{B}_{\max}(\psi) \leq \mathcal{B}_{\max}(\psi_m)$, when $\mathcal{G}(\psi) = \mathcal{G}(\psi_m)$. Hence, for the states in $GHZ^R$ and $W$ classes, one can show the following complementary relation
\begin{align}
4\mathcal{G}(\psi)\big(1-\mathcal{G}(\psi)\big)+\mathcal{B}_{\max}(\psi) \leq 1.
\end{align} 

 We study the complementary relation between the GGM and the maximum Bell inequality violation for Haar uniformly generated random genuinely entangled pure three qubit states in Fig.\ref{fig:ggm-bv-ghz-w}. Numerical study (appendix \ref{numerics}) shows that it holds for all the $GHZ$ class states. Therefore, we propose that the following complementary relation holds for all three qubit pure states.
\begin{claim}
For any three qubit pure state $\ket{\psi}$ the GGM, $\mathcal{G}(\psi)$ and the maximal bipartite Bell inequality violation, $\mathcal{B}_{\max}(\psi)$ obey the following complementary relation
\begin{equation}
4\mathcal{G}(\psi)\big(1-\mathcal{G}(\psi)\big)+\mathcal{B}_{\max}(\psi) \leq 1.
\label{corollary1}
\end{equation}
\end{claim}

\subsection{Discord Monogamy Score versus Maximum Bell Inequality Violation}
\begin{figure}[tbhp!]
\begin{center}
\includegraphics[scale=0.57]{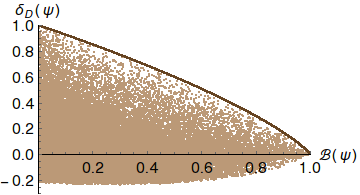}
\end{center}
\caption{Complementarity between the discord monogamy score,  $\delta_D(\psi)$ and the maximum Bell inequality violation, $\mathcal{B}_{\max}(\psi)$ for $10^6$ number of Haar uniformly generated random three qubit pure states. The MBV states lie at the boundary. Both axes are dimensionless.}
\label{fig:dms-bv-ghz-w}
\end{figure}
 Moreover, we study the complementary relation between the discord monogamy score (DMS), and the bipartite Bell inequality violation numerically (see Fig.\ref{fig:dms-bv-ghz-w}). 
The DMS is not always non-negative as the discord is not monogamous. The MBV states lie at the boundary of the complementary relation in this case too.

\section{Conclusions}
\label{conclu}
 Quantifying nonlocality and finding a relation between nonlocality and quantum correlations is an important 
yet challenging task in multipartite and higher dimensional quantum systems. In recent past it has been observed that 
correlation statistics of two body systems can be very fruitful in inferring the multipartite properties of a composite 
quantum system. In this article we find a single parameter class of states called the MBV class that exhibits the 
maximum Bell inequality violation for reduced two qubit systems for a fixed amount of genuine tripartite correlation. 
The measures that we have chosen to quantify genuine tripartite correlation, belong to both the 
entanglement separability paradigm  (the tangle and the GGM) as well as the information 
theoretic paradigm (DMS).
This result in turn establishes a complementary relation between the Bell-CHSH violation by the reduced bipartite systems 
and the genuine quantum correlations (the tangle) in all three qubit states, with the MBV class lying at the boundary of the 
complementary relation. The complementary relation between the Bell-CHSH violation by the reduced bipartite systems 
and the genuine quantum correlations  also holds for genuine correlation measure the GGM and the DMS in three qubit pure states, with the MBV class lying at the boundary of the 
complementary relation.
The complementary relation suggests that for all the three measures, the Bell inequality violation of the reduced two party system comes at the cost of genuine tripartite correlations present in the whole system.


\begin{acknowledgments}
 AM thanks Dagmar Bruss, Michael Horne, Marcus Huber, Hermann Kampermann, Miguel Navascues, Ujjwal Sen and Jochen Szangolies for fruitful discussions and suggestions. PP acknowledges IIIT-Hyderabad for support during the MS project. AM acknowledges IIIT-Hyderabad for hospitality and support during academic visit and research fellowship of DAE, Govt. of India.
\end{acknowledgments}

\appendix
\section{Expressions for tangle}
\label{app:tangle}
 For a three qubit pure state $\ket{\psi}_{ABC} = \sum_{i,j,k=0}^1 c_{ijk} \ket{ijk}$, the 3-tangle or the residual tangle \cite{Coffmankunduwotters} is given as
\begin{align}
\tau_{ABC} &= \tau_{A:BC}-\tau_{AB}-\tau_{AC}, \nonumber\\
		   &= 4\left\vert s_1 -2s_2+4s_3\right \vert \label{eq:tangle-def},
\end{align}
where,
\begin{align}
s_1& = \cc{000}\cc{111}+\cc{001}\cc{110}+\cc{010}\cc{101}+\cc{100}\cc{011}, \\
s_2& = \c{000}\c{111}\left ( \c{011}\c{100}+\c{101}\c{010}+\c{110}\c{001} \right ) \nonumber \\
      &+\c{011}\c{100}\left ( \c{101}\c{010}+\c{110}\c{001} \right ) \nonumber \\
      &+\c{101}\c{010}\c{110}\c{001}, \\
s_3& = \c{000}\c{110}\c{101}\c{011} + \c{111}\c{001}\c{010}\c{100}.
\end{align}

\subsection{GHZ Class}
For $\ket{\psi}_{GHZ} = \sqrt{K} \left( c_{\d} \ket{000} + s_{\d} e^{i \phi} \ket{\varphi_\a}\ket{\varphi_\b}\ket{\varphi_\d} \right)$ (Eq.\eqref{ghz-class-state}), $s_1$, $s_2$ and $s_3$ are respectively given as
\begin{align}
s_1 & = e^{2i\phi} t_2 \frac{\left ( c_\d^2 + 2 e^{2i \phi} t_1\right )}{(1+\mathcal{X}_\phi)^2}, \\
s_2 & = 3e^{3i\phi} t_1 t_2 \frac{\left ( c_\d+2e^{i\phi}t_1\right )}{(1+\mathcal{X}_\phi)^2}, \\
s_3 & = e^{3i\phi} t_1 t_2 \frac{\left ( c_\d+2e^{i\phi}t_1\right )}{(1+\mathcal{X}_\phi)^2},  
\end{align}
where, $t_1 = c_\a c_\b c_\g s_\d$, $t_2 = s_\a^2s_\b^2s_\g^2s_\d^2$ and $\mathcal{X}_\phi = c_\a c_\b c_\g c_\phi s_{2\d}$. 
%
 Hence, the tangle of $\ket{\psi}_{GHZ}$ is given as
\begin{equation}
\label{eq:tangle-ghz}
\tau_{GHZ} = \frac{s^2_\a s^2_\b s^2_\g s^2_{2\d}}{(1+c_\a c_\b c_\g c_\phi s_{2\d})^2}
\end{equation}


\subsection{Maximally Bell inequality violating (MBV) class}
 For the MBV states $\ket{\psi_\m} = \frac{\ket{000} + \m(\ket{010}+\ket{101}) + \ket{111}}{\sqrt{2+2\m^2}}$ (Eq.\eqref{mdcc-class-state}), $s_1$, $s_2$ and $s_3$ are respectively given by
\begin{align}
s_1 &= \frac{1+m^4}{4(1+m^2)^2}, \\
s_2 &= \frac{m^2}{4(1+m^2)^2}, \\
s_3 &= 0,
\end{align}
and the tangle is,
\begin{equation}
\label{eq:tangle-mdcc}
\tau_m = 1 - \frac{4m^2}{(1+m^2)^2}.
\end{equation}

\section{Expressions for GGM}
\label{app:ggm}
 The GGM for a three qubit pure state $\ket{\psi}_{ABC}$, is calculated as
\begin{equation*}
\mathcal{G}(\psi_{ABC}) = 1 - \max \{\lambda_A,\lambda_B,\lambda_C\},
\end{equation*}
where $\lambda_A$, $\lambda_B$ and $\lambda_C$ are the largest eigenvalues of the reduced systems $\rho_A$, $\rho_B$ and $\rho_C$ respectively. 
\subsection{GHZ Class}
 For $\ket{\psi}_{GHZ}$, the eigenvalues of the reduced systems $\rho_A$, $\rho_B$ and $\rho_C$ are respectively given as follows.\\
For subsystem $\rho_A$,
\begin{align*}
g_A^\pm = \frac{1}{2}\left (1 \pm \sqrt{1+\frac{ s_\a^2(c_\b^2 c_\g^2 -1)  s_{2\d}^2}{ (1+\mathcal{X_\phi})^2 }}\right ) ,
\end{align*}
For subsystem $\rho_B$,
\begin{align*}
g_B^\pm = \frac{1}{2}\left (1 \pm \sqrt{1+\frac{ s_\b^2(c_\a^2 c_\g^2 -1)  s_{2\d}^2}{ (1+\mathcal{X_\phi})^2 } }\right ) ,
\end{align*}
For subsystem $\rho_C$,
\begin{align*}
g_C^\pm = \frac{1}{2}\left (1 \pm \sqrt{1+\frac{ s_\g^2(c_\a^2 c_\b^2 -1)  s_{2\d}^2}{ (1+\mathcal{X_\phi})^2 } }\right ) .
\end{align*}
%
Here $\mathcal{X_\phi}=c_\a c_\b c_\g c_\phi s_{2\d}$. 
In each of the subsystem, it is clear that the eigenvalue $g^+_X$ is largest, and will be relabeled $g_X$, where $X\in\{A, B, C\}$. 
\subsection{W Class}
 For $\ket{\psi}_W = \sqrt{d}\ket{000} + \sqrt{a}\ket{001} + \sqrt{b}\ket{010} + \sqrt{c}\ket{100}$, 
the eigenvalues of the reduced subsystems are as follows. For subsystem $\rho_A$,
\begin{align*}
w_A^\pm = \frac{1}{2}\left (1 \pm \sqrt{1- 4(a+b)c}\right ) ,
\end{align*}
For subsystem $\rho_B$,
\begin{align*}
w_B^\pm = \frac{1}{2}\left (1 \pm \sqrt{1- 4(a+c)b}\right ) ,
\end{align*}
For subsystem $\rho_C$,
\begin{align*}
w_C^\pm = \frac{1}{2}\left (1 \pm \sqrt{1- 4(b+c)a}\right ) ,
\end{align*}
Again, the eigenvalue $w^+_X$ in each pair is largest of the two, where $X\in\{A, B, C\}$.
\subsection{MBV class}
 For $\ket \psi_m$, the eigenvalues for $\rho_A$ and $\rho_C$ are $\{1/2,1/2\}$, and eigenvalues of $\rho_B$ are given by,
\begin{align*}
\lambda_B^\pm=\frac{1}{2}\pm\frac{m}{1+m^2} .
\end{align*}
Clearly, $\lambda_B^+$ is the greatest among all the eigenvalues. Hence, the GGM is given by
\begin{equation}
\label{eq:ggm-mdcc}
\mathcal{G}(\psi_m) = 1 - \lambda_B^+ = \frac{1}{2}-\frac{m}{1+m^2}.
\end{equation}

\section{Expressions for the Bell inequality violation}
\label{app:bv}
 Here we give the analytical expressions for the Bell inequality violations for $GHZ^R$ and $W$ classes of states, calculated as per the method prescribed by Horodecki \textit{et al.} \cite{horodecki-mrho}. The Bell inequality violation is quantified as $\mathcal{B}(\rho_{AB}) = \max\{0,M(\rho)-1\}$, where $M(\rho_{AB})$ is the sum of the two largest eignevalues of the matrix, $T_{AB}^TT_{AB}$. Here, $T_{AB}$ is the correlation matrix such that $(T_{AB})_{ij} = \mbox{Tr}[(\sigma_i\otimes\sigma_j)\rho_{AB}]$.
As we are dealing with three qubit pure states, the Bell inequality violation is calculated for each of the three reduced two qubit systems, $\rho_{AB}$, $\rho_{BC}$ and $\rho_{AC}$.

\subsection{$GHZ^R$ Class}
 For the state $\ket{\psi}_{GHZ^R}$, the eigenvalues of the matrix $T^T_{BC}T_{BC}$ are as follows,
\begin{align}
\lambda_1 &= \frac{c_\a^2 s_\b^2s_\g^2s_\d^2}{(1+\mathcal{X})^2} ,\\
\lambda_2 &= \mathcal{A}_{BC} - \frac{\sqrt{\mathcal{D}}}{2(1+\mathcal{X})^2} , \\
\lambda_3 &= \mathcal{A}_{BC} + \frac{\sqrt{\mathcal{D}}}{2(1+\mathcal{X})^2} ,
\end{align}
where, $\mathcal{X} = c_\a c_\b c_\g c_{2\d}$, and 
\begin{equation}
\mathcal{A}_{BC} = \frac{1}{2}+\frac{(c_\a^2 -c_\b^2 -c_\g^2 + 2c_\b^2 c_\g^2) s^2_{2 \d } -\mathcal{X}^2 }{ 2(1+\mathcal{X})^2}, 
\end{equation}
\begin{equation}
\label{eq:d-term}
\begin{split}
\mathcal{D} = 1 +& 4\mathcal{X} + s_{2\d}^2\big[ -s_{2\d}^2 \sum c_\a +( 4\mathcal{X} - 2 c_{2\d}^2 )\sum c_\a^2 \\&+ 2 (1+c_{2\d}^2)\sum c_\a^2c_\b^2 -8\mathcal{X} + 4\mathcal{X}^2\big],
\end{split}
\end{equation}
with $\sum c_\a = (c_\a+c_\b+c_\g)$, $\sum c_\a^2 = (c_\a^2+c_\b^2+c_\g^2)$ and $\sum c_\a^2c_\b^2 = (c_\a^2c_\b^2 +c_\a^2c_\g^2+ c_\b^2c_\g^2)$.

 It can be shown that the eigenvalues are in the following order, $\lambda_1<\lambda_2<\lambda_3$. Thus, the Bell-CHSH value, $M(\rho_{BC}) = \lambda_3+\lambda_2 = 2\mathcal{A}_{BC}$,
\begin{equation}
M(\rho_{BC}) = 1+\frac{(c_\a^2 -c_\b^2 -c_\g^2 + 2c_\b^2 c_\g^2) s^2_{2 \d } -\mathcal{X}^2 }{ (1+\mathcal{X})^2} .
\label{eq:ghz-bv-bc}
\end{equation}

 Similarly, for both the subsystems $\rho_{AC}$ and $\rho_{AB}$, the eigenvalues of $T^T_{AC}T_{AC}$ and $T^T_{AB}T_{AB}$ respectively, follow the ordering $\lambda_1<\lambda_2<\lambda_3$. The corresponding Bell-CHSH values are given as
\begin{equation}
M(\rho_{AC}) = 1+\frac{(c_\b^2 -c_\a^2 -c_\g^2 + 2c_\a^2 c_\g^2) s^2_{2 \d } -\mathcal{X}^2 }{ (1+\mathcal{X})^2} .
\label{eq:ghz-bv-ac} 
\end{equation}
\begin{equation}
M(\rho_{AB}) = 1+\frac{(c_\g^2 -c_\a^2 -c_\b^2 + 2c_\a^2 c_\b^2) s^2_{2 \d }-\mathcal{X}^2 }{ (1+\mathcal{X})^2} .
\label{eq:ghz-bv-ab}
\end{equation}
 There is a inherent exchange symmetry in the expressions, where in, if we swap parameters $\a$ and $\b$ in $M(\rho_{BC})$ we obtain $M(\rho_{AC})$. Similarly, swapping $\b$ and $\g$ in $M(\rho_{AC})$ gives $M(\rho_{AB})$.
 For each subsystem, we define the Bell inequality violation as  $\mathcal{B}(\rho_{XY}) = \max\{0,M(\rho_{XY})-1\}$.


\subsection{W Class}
 For the state $\ket{\psi}_{W}$, 
the eigenvalues for the matrix $T_{BC}^TT_{BC}$ are given as,
\begin{align}
\lambda_1 &= 4ab ,\\
\lambda_2 &= \frac{1}{2} + 2(ab - ac -bc) - \frac{\sqrt{V}}{2} ,\\
\lambda_3 &= \frac{1}{2} + 2(ab - ac -bc) + \frac{\sqrt{V}}{2} ,
\end{align}
where $a+b+c+d=1$, and $V = [(\sqrt{a}+\sqrt{b}-\sqrt{c})^2+d] [(\sqrt{a}-\sqrt{b}+\sqrt{c})^2+d] [(-\sqrt{a}+\sqrt{b}+\sqrt{c})^2 +d] [(\sqrt{a}+\sqrt{b}+\sqrt{c})^2+d]$.

 The eigenvalues follow the ordering, $\lambda_2<\lambda_1<\lambda_3$, implying $M(\rho_{BC}) = \lambda_1+\lambda_3$,
\begin{equation}
M(\rho_{BC}) = \frac{1}{2} + \frac{ 12a b -4 a c - 4 b c + \sqrt{V}}{ 2}.
\label{eq:w-bv-bc}
\end{equation}

 Similarly, for the subsystems $\rho_{AC}$ and $\rho_{AB}$, the eigenvalues for $T_{AC}^TT_{AC}$  and $T_{AB}^TT_{AB}$ follow the same ordering as mentioned earlier and hence, the Bell-CHSH values for both the bipartitions are given by 
\begin{align}
M(\rho_{AC}) &= \frac{1}{2} + \frac{ -4a b +12 a c - 4 b c + \sqrt{V}}{ 2 },
\label{eq:w-bv-ac}
\\
M(\rho_{AB}) &= \frac{1}{2} + \frac{ -4a b -4 a c + 12 b c + \sqrt{V}}{ 2} .
\label{eq:w-bv-ab}
\end{align}
%
%

 Note that, here too an exchange symmetry similar to the $GHZ^R$ class of states holds. By swapping $a$ and $b$ in $M(\rho_{AB})$, we get $M(\rho_{AC})$. Then by swapping $b$ and $c$ we obtain $M(\rho_{BC})$ from $M(\rho_{AC})$. The Bell inequality violations for each subsystem are then obtained as, $\mathcal{B}(\rho_{XY}) = \max\{0,M(\rho_{XY})-1\}$.

\subsection{MBV class}

 In case of the MBV class state $\ket{\psi}_m$, the correlation matrices for $\rho_{BC}$ and $\rho_{AB}$ are equal and hence the expressions for the Bell-CHSH values are equal. Therefore, only $\rho_{AC}$ can exhibit violation of the Bell Inequality.
 The eigenvalues for $T_{AC}^TT_{AC}$ for $\rho_{AC}$ are then,
\begin{align}
\lambda_1 &= 1 ,\\
\lambda_2 &= \lambda_3 = \frac{4m^2}{(1+m^2)^2}.
\end{align}
This implies that
\begin{equation}
\label{eq:mdcc-rhoAC}
M(\rho_{AC})=1+\frac{4m^2}{(1+m^2)^2}.
\end{equation}
Therefore,
\begin{equation}
\label{eq:bellviolation-mdcc}
M(\psi_m) = 1+\frac{4m^2}{(1+m^2)^2}.
\end{equation}

\section{Numerical Method}
\label{numerics}
 To perform the numerical study we have generated Haar-uniformly distributed random three qubit pure states. As the sets of fully separable, bi-separable and the ${\it W}$ class states have measure zero with respect to the set of ${\it GHZ}$ class states, almost all of the Haar-uniformly generated random pure states belong to the ${\it GHZ}$ class. This has been cross checked by calculating the non vanishing tangle for the states generated in the aforesaid way. Now for each such randomly generated state we have evaluated the maximum bipartite Bell inequality violation and the three correlation measures, namely the tangle, the GGM and the DMS. We have performed our study for $10^7$ number of randomly generated states for each measure. However, the plots exhibit the numerical study for $10^6$ number of states.


\end{document}